\documentclass{eptcs}

\usepackage{amsmath}
\usepackage{amssymb}
\usepackage{url}
\usepackage{amsthm}
\usepackage{graphicx}
\usepackage{amscd}
\usepackage[usenames,dvipsnames,svgnames,table]{xcolor}

\usepackage{tikz}
\usetikzlibrary{positioning,decorations.markings,circuits,decorations.pathmorphing,decorations.pathreplacing,arrows,automata,}
\newcommand{\Run}{\operatorname{Run}}
\newcommand{\N}{\mathbb{N}}

\begin{document}

\theoremstyle{definition}
\newtheorem{theorem}{Theorem}
\newtheorem{definition}[theorem]{Definition}
\newtheorem{problem}[theorem]{Problem}
\newtheorem{assumption}[theorem]{Assumption}
\newtheorem{corollary}[theorem]{Corollary}
\newtheorem{proposition}[theorem]{Proposition}
\newtheorem{example}[theorem]{Example}
\newtheorem{lemma}[theorem]{Lemma}
\newtheorem{observation}[theorem]{Observation}
\newtheorem{fact}[theorem]{Fact}
\newtheorem{question}[theorem]{Open Question}
\newtheorem{conjecture}[theorem]{Conjecture}
\newtheorem{addendum}[theorem]{Addendum}
\newcommand{\uint}{{[0, 1]}}
\newcommand{\Cantor}{{\{0,1\}^\mathbb{N}}}
\newcommand{\name}[1]{\textsc{#1}}
\newcommand{\me}{\name{P.}}
\newcommand{\id}{\textrm{id}}
\newcommand{\dom}{\operatorname{dom}}
\newcommand{\Dom}{\operatorname{Dom}}
\newcommand{\codom}{\operatorname{CDom}}
\newcommand{\spec}{\operatorname{spec}}
\newcommand{\opti}{\operatorname{Opti}}
\newcommand{\optis}{\operatorname{Opti}_s}
\newcommand{\Baire}{\mathbb{N}^\mathbb{N}}
\newcommand{\hide}[1]{}
\newcommand{\mto}{\rightrightarrows}
\newcommand{\Sierp}{Sierpi\'nski }
\newcommand{\BC}{\mathcal{B}}
\newcommand{\C}{\textrm{C}}
\newcommand{\lpo}{\textrm{LPO}}
\newcommand{\llpo}{\textrm{LLPO}}
\newcommand{\bwt}{\textrm{BWT}}
\newcommand{\leqW}{\leq_{\textrm{W}}}
\newcommand{\leW}{<_{\textrm{W}}}
\newcommand{\equivW}{\equiv_{\textrm{W}}}
\newcommand{\equivT}{\equiv_{\textrm{T}}}
\newcommand{\geqW}{\geq_{\textrm{W}}}
\newcommand{\pipeW}{|_{\textrm{W}}}
\newcommand{\nleqW}{\nleq_\textrm{W}}
\newcommand{\leqsW}{\leq_{\textrm{sW}}}
\newcommand{\equivsW}{\equiv_{\textrm{sW}}}
\newcommand{\Sort}{\operatorname{Sort}}

\newcommand{\pitc}{\Pi^0_2\textrm{C}}

\title{Parameterized Games and Parameterized Automata}

\author{
Arno Pauly
\institute{Department of Computer Science\\Swansea University\\ Swansea, UK\\}
\email{Arno.M.Pauly@gmail.com}
}

\def\titlerunning{Parameterized Games and Parameterized Automata}
\def\authorrunning{A. Pauly}
\maketitle

\begin{abstract}
We introduce a way to parameterize automata and games on finite graphs with natural numbers. The parameters are accessed essentially by allowing counting down from the parameter value to $0$ and branching depending on whether $0$ has been reached. The main technical result is that in games, a player can win for some values of the parameters at all, if she can win for some values below an exponential bound. For many winning conditions, this implies decidability of any statements about a player being able to win with arbitrary quantification over the parameter values.

While the result seems broadly applicable, a specific motivation comes from the study of chains of strategies in games. Chains of games were recently suggested as a means to define a rationality notion based on dominance that works well with quantitative games by Bassett, Jecker, P., Raskin and Van den Boogard. From the main result of this paper, we obtain generalizations of their decidability results with much simpler proofs.

As both a core technical notion in the proof of the main result, and as a notion of potential independent interest, we look at boolean functions defined via graph game forms. Graph game forms have properties akin to monotone circuits, albeit are more concise. We raise some open questions regarding how concise they are exactly, which have a flavour similar to circuit complexity. Answers to these questions could improve the bounds in the main theorem.
\end{abstract}

\section{Introduction}
The study of various kinds of $(\omega)$-automata and of games played on finite graphs tend to go hand in hand: In one direction, universality of a non-deterministic automaton can be reduced to asking about a winning strategy by Player 1 in a game where Player 2 controls the input, and Player 1 controls the non-determinism. In the other direction, we can view any $(\omega)$-automaton $\mathcal{M}$ as deciding the winning condition of a class of games (with varying arenas).

Our goal here is to study games and automata that are enhanced by additional natural number parameters. We will suggest a means for how these parameters are accessed that is sufficiently powerful to express meaningful concepts, yet keeps the usual algorithmic questions decidable no matter how we quantify over the parameters. This approach should be broadly applicable in many areas using games and automata.

The immediate motivation comes from \cite{ba-arxiv}, where parameterized automata were used to define chains of strategies with respect to dominance. Consider the game depicted in Figure \ref{fig:best_animal}. The type of behaviour we wish to describe is \emph{Repeat the $v_0v_1$ loop $k$ times (unless $\ell_2$ is reached), then move to $\ell_1$}. For any concrete choice of $k$, this is suboptimal, for choosing larger $k$ would be an improvement. However, it is not better to repeat the loop forever, since that risks getting payoff $0$. We thus consider this as a single instance of strategic behaviour, parameterized by $k$.

\begin{figure}[h]
\centering
    \begin{tikzpicture}[->,>=stealth',shorten >=1pt, initial text={}]
      \node [initial above ,state] (q1)                      {$v_0$};
      \node[state,rectangle]          (q2) [right=of q1]         {$v_1$};
      \node[state]          (q3) [right=of q2]         {$\ell_2$};
     \node[state]          (q4) [left=of q1]         {$\ell_1$};
      \path (q1) edge [bend left] node {$$} (q2);
      \path (q2) edge [bend left] node {$$} (q1);
      \path (q2) edge node {$$} (q3);
      \path (q1) edge node {$$} (q4);
       \path (q3) edge [loop above] node {$$} (q3);
            \path (q4) edge [loop above] node {$$} (q4);
    \end{tikzpicture}
    \caption{The \emph{Help-me?}-game from \cite[Example 1]{ba-arxiv}. The protagonist owns the circle vertices.
The payoffs are defined as follows: $p((v_0v_1)^{\omega}) = 0$, $p((v_0v_1)^n v_0 \ell_1^{\omega}) = 1$ for $n \in \N$ and $p((v_0v_1)^n \ell_2^{\omega}) = 2$ for $n \in \N$. The protagonist seeks to maximize payoff.}
\label{fig:best_animal}
\end{figure}
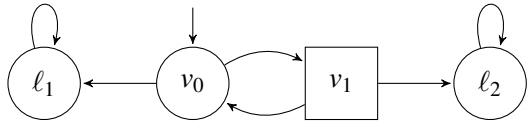

To be able to reason about such parameterized strategies in an algorithmic way, we will lift the usual decidability-aspects of finite automata or even $\omega$-automata to parameterized automata. We do this by showing that any possible behaviour for some parameter value in a parameterized game is already exhibited for bounded parameter values (Theorem \ref{theo:main}), entirely independent of the type of prefix-independent winning condition used. We will detail the application to chains of strategies in Section \ref{sec:application-chains}. The principles how this can be applied will be laid out in Section \ref{sec:application-generic}.

Along the way to proving the main result, we end up exploring the notion of functions definable by games in Section \ref{sec:ggf}. This can be seen as a more concise expression of monotone boolean functions than provided by monotone circuits. Obtaining a better understanding of how expressive this notion is precisely would provide either improved bounds in our main theorem, or lower bounds ruling them out. It seems very plausible that adapting and extending techniques from circuit complexity could be crucial here.

\subsection{Potentially related work}
A similar theme to our work is found in the study of regular cost functions. There, too, an extension of the usual automata and logical framework is undertaking in order to deal with bounds on certain quantities. Besides the shared theme, there does not seem to be significant overlap. An introduction to regular cost functions is found in \cite{colcombet}.

Parameterized games have some (superficial?) similarity to multi-dimensional energy games (e.g.~\cite{chatterjee,jurdzinski2}), with the restriction that all weights are non-positive. This restriction of energy games, however, is not particularly interesting (unless combined with other features such as in consumption games \cite{novotny}). The main difference is that in an energy game, depleted energy levels would force Player 1 to chose the non-costly edge, if she controls the relevant vertex, but she would be allowed to take the same edge at any time. In a parameterized game, no player decides on how to act in a counter access vertex, it is completely determined by the current counter value. Apart from simple examples, this makes their analysis very different.

\section{Introducing parameterized games}

\begin{definition}
A $N$-\emph{parameterized arena} is a directed graph $G = (V,E)$ together with a partition $V = V_1 \cup V_2 \cup V_c$, and a further partition $V_c = V_c^1 \cup \ldots \cup \ldots V_c^N$, where each vertex $v \in V_c$ has two distinguishable outgoing edges (a \emph{red} and a \emph{green} edge), and each $v \in V_1 \cup V_2$ has at least one outgoing edge. A vertex has either no self-loop, or a self-loop and otherwise no outgoing edges (those vertices are called \emph{leaves}).
\end{definition}

Arenas serve as the setting for games, where a token is moved along the edges of the graph. There are two players, one who decides which outgoing edge to take at vertices $v \in V_1$, and one who decides at $v \in V_2$. Our new addition are the vertices $v \in V_c$. The behaviour there is controlled via $N$ counters. The counters are initialized to some values $(n^1,\ldots,n^N) \in \mathbb{N}$. Whenever the token reaches some $v \in V_c^j$, we check whether the $j$-th counter value is $0$. If yes, the token moves along the red edge. If no, we decrement the $j$-th counter and take the green edge.

Formally, let a \emph{strategy} be a function $\sigma : V^* \to V$ satisfying that if $\sigma(hv) = u$, then $(v,u) \in E$. Given a strategy $\sigma_i$ for each player, a starting vertex $v_0$ and an initial counter value $n_0^j$ for each $j \leq N$, we inductively define the \emph{induced run} $\Run(\sigma_1,\sigma_2,v_0,(n_0^1,\ldots,n_0^N)) = v_0v_1v_2\ldots$ and the counter value updates $n_0^j, n_1^j, \ldots$ in stages $k \in \mathbb{N}$ as follows: If $v_k \in V_i$, then $n_{k+1}^j = n_k^j$ and $v_{k+1} = \sigma_i(v_0\ldots v_k)$. If $v_k \in V_c^j$ and $n_k^j = 0$, then $v_{k+1} = u$ where $u$ is the vertex reached by following the red edge from $v_k$, and $n_{k+1}^j = 0$, as well as $n_{k+1}^{j'} = n_{k}^{j'}$ for $j' \neq j$. If $v_k \in V_c^j$ and $n_k^j > 0$, then $v_{k+1} = u'$ where $u'$ is reached following the green edge from $v_k$, and $n_{k+1}^j = n_k^j - 1$, as well as $n_{k+1}^{j'} = n_{k}^{j'}$ for $j' \neq j$.

\begin{definition}
A \emph{parameterized game} is a parameterized arena together with a winning condition $C \subseteq V^\omega$. We only consider prefix-independent winning conditions here; these satisfy that $\forall \rho \in V^\omega, \ \forall h \in V^* \ \ \rho \in C \Leftrightarrow h\rho \in C$.
\end{definition}

\begin{definition}
A strategy $\sigma_1$ is \emph{winning} a parameterized game for Player 1 from $v_0$ with parameters $n_0^1, \ldots, n^N$, if for all strategies $\sigma_2$ it holds that $\Run(\sigma_1,\sigma_2,v_0,(n_0^1,\ldots,n_0^N)) \in C$.
\end{definition}

Of course, we have the dual notion of a strategy being winning for Player 2. Whenever $C \subseteq V^\omega$ is a Borel set, then Borel determinacy \cite{martin} implies that either Player 1 or Player 2 has a winning strategy in a parameterized game for each choice of parameters $n_0^1, \ldots, n^N$. For sufficiently complex winning conditions, this property no longer holds -- being prefix-independent does not help (cf.~Proposition \ref{prop:undetermined} in the appendix).

Considering parameterized games for particular choices of parameters is not particularly interesting:

\begin{observation}
Given a parameterized game $G$ of size $S$ and parameter values $n_0^1,\ldots,n_0^N$, we can unfold it to an ordinary game $G'$ of size $S\prod_{i \leq N}n_0^i$ such that Player $1$ can win $G$ with parameter values $n_0^1,\ldots,n_0^N$ iff she can win $G'$.
\end{observation}

Instead, we are interested in how whether Player 1 has a winning strategy in a fixed game varies with the choice of parameters $(n_0^1,\ldots,n_0^N)$. Our main theorem shows that in order to explore the possible behaviour for all parameter values, it it suffices to check a finite number of cases, namely:

\begin{theorem}
\label{theo:main}
Fix a parameterized game $((V_0 \cup V_1 \cup V_c, E),C)$ and starting vertex $v_0$. Player 1 has a winning strategy from $v_0$ for some parameter values $(n_0^1,\ldots,n_0^N)$ iff there exists such $(n_0^1,\ldots,n_0^N)$ with $n_0^j \leq 2^{|V_c^j|}$.
\end{theorem}

Our main theorem immediately implies that if the existence of a winning strategy for some class of games is decidable, then it is still decidable for its parameterized version -- albeit potentially at the cost of an exponential blowup of the computation time.

\begin{example}
\label{ex:linearlowerbound}
There are single-player parameterized reachability games with a single parameter of size $N + 2$, such that the player wins iff the parameter is at least $N$. Simple construct a line of counter access states linked by green edges, with the single winning leaf at the end. The red edges all lead to a losing leaf. See Figure \ref{fig:linearlowerbound} for an example.
\end{example}

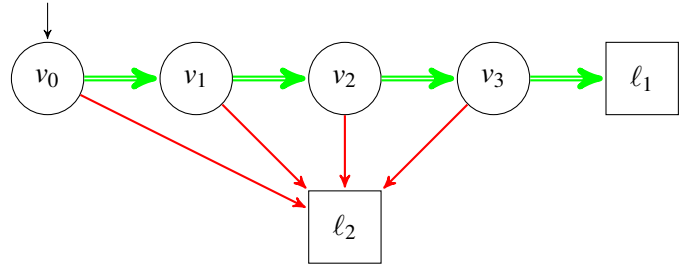
\begin{figure}[h!]
\centering
    \begin{tikzpicture}[->,>=stealth',shorten >=1pt, initial text={}]
      \node [initial above ,state] (q1)                      {$v_0$};
      \node[state]          (q2) [right=of q1]         {$v_1$};
      \node[state]          (q3) [right=of q2]         {$v_2$};
     \node[state]          (q4) [right=of q3]         {$v_3$};
          \node[state,rectangle]          (l1) [right=of q4]         {$\ell_1$};
           \node[state,rectangle]          (l2) [below=of q3]         {$\ell_2$};
      \path (q1) edge [green, double, thick] node {$ $} (q2);
      \path (q2) edge [green, double, thick] node {$$} (q3);
      \path (q3) edge [green, double, thick] node {$$} (q4);
        \path (q4) edge [green, double, thick] node {$$} (l1);
             \path (q1) edge [red, thick] node {$$} (l2);
      \path (q2) edge [red, thick] node {$$} (l2);
            \path (q3) edge [red, thick] node {$$} (l2);
                  \path (q4) edge [red, thick] node {$$} (l2);
    \end{tikzpicture}
    \caption{To Example \ref{ex:linearlowerbound}. A single-player parameterized reachability game with a single parameter: The leaf $\ell_1$ is winning for the player, the leaf $\ell_2$ is losing. The player never gets to chose, and wins iff the parameter at the beginning is $4$ or greater.}

    \label{fig:linearlowerbound}
\end{figure}

As a variant\footnote{Which was suggested to the author by one of the anonymous referees.}, we could turn parameterized games (which are families of games) into a single game model satisfying the same purposes. For that, rather than working with fixed-but-unspecified parameters, we add initial rounds where the players get to chose their desired values of parameters (with some assignment of which player gets to pick what parameters). Theorem \ref{theo:main} then shows that the infinitary rounds at the start can be replaced by finite choices without altering who is winning.

\section{Graph game form definable functions}
\label{sec:ggf}
We proceed to introduce the main technical notion for improving the bounds in our main theorem. We consider functions of type $f : \{0,1\}^m \to \{0,1\}^n$ that can be defined in particular way via games. Besides having the key role in the main proof, this notion might also be of independent interest. Defining functions by games can be seen as a generalization or extension of the very familiar concept of defining them via circuits.

\begin{definition}
\label{def:graphgameform}
An $(m,n)$-\emph{graph game form}\footnote{The concept and name are inspired by the notion of a game form as related to normal form games. The latter was introduced in \cite{gibbard}, and has since seen significant attention in game theory.} is a parameterized game $(G,C)$ with $m$ designated leaves $(\ell_1,\ldots,\ell_m)$ and $n$ designated vertices $(s_1,\ldots,s_n)$; instantiated with parameter values. It defines a function $f : \{0,1\}^m \to \{0,1\}^n$ as follows: Write $a_1,\ldots,a_n = f(b_0,\ldots,b_m)$. Consider the game with the modified winning condition $C' = C \setminus \{hl_i^\omega \mid h \in V^*, b_i = 0\} \cup \{hl_i^\omega \mid h \in V^*, b_i = 1\}$. If Player 1 has a winning strategy in this game starting from $s_j$, then $a_j = 1$, otherwise $a_j = 0$.
\end{definition}

In other words, the input to $f$ tells us what happens in the game if the run ever reaches one of the designated leaves -- either Player 1 wins there or not. For those runs never reaching a designated leaf, we stick with the original winning condition $C$ to determine the winner. We then vary the starting vertices amongst the designated choices $s_1,\ldots,s_n$ to determine the output bits.


We will proceed to see that the class of functions definable via graph game forms is a very familiar class of boolean functions, namely the monotone functions (see \cite{korshunov} for a survey). A function $f : \{0,1\}^m \to \{0,1\}^n$ is called \emph{monotone}, if $u \leq w$ implies $f(u) \leq f(w)$. Here, we lift $\leq$ from $\{0,1\}$ to $\{0,1\}^m$ componentwise.

Recall that a monotone circuit is an acyclic directed graph where vertices that are not sinks are labeled by $\wedge$ or $\vee$. If it has $m$ sinks and $n$ sources, it computes a function $f : \{0,1\}^m \to \{0,1\}^n$ by assigning \emph{true} or \emph{false} to the sinks based on the input bits, assigning a $\wedge$-labeled vertex to \emph{true} if all its successors are assigned \emph{true}, and \emph{false} if one successor is assigned \emph{false}, and dually for $\vee$-labeled vertices. The output bits are obtained by considering the values assigned to the sources. In addition, some output bits could be fixed to be either $0$ or $1$; and some input bits could be entirely ignored.

It is a classic observation that the monotone boolean functions are exactly those computed by monotone circuits. Straight-forward induction shows that being computed by a monotone circuit implies being monotone. For the other direction, we note that every boolean function $f$ can be expressed bitwise as a reduced disjunctive normal form. If this contains a negation, we can extract a counterexample to $f$ being monotone. If it does not, we can directly transform the disjunctive normal form into a monotone circuit.

\begin{theorem}
\label{theo:monotonefunctions}
The following are equivalent for boolean functions $f : \{0,1\}^m \to \{0,1\}^n$:
 \begin{enumerate}
 \item $f$ is monotone.
 \item $f$ is definable via a reachability graph game form.
 \item $f$ is definable via some graph game form.
 \end{enumerate}
\begin{proof}
\hfill
\begin{description}
\item[$1. \Rightarrow 2.$]
We use the characterization of monotone functions as being computed by monotone circuits. It is straight-forward to conceive of a monotone circuit as a graph game form with some arbitrary winning condition. We let Player 1 control the $\vee$-labeled vertices, and Player 2 the $\wedge$-labeled vertices. Since a circuit is acyclic, any path eventually reaches one of the designated leaves, and thus the original winning condition has no impact at all.

\item[$2. \Rightarrow 3.$] Trivial.

\item[$3. \Rightarrow 1.$] Let $f$ be defined by a graph game form. To show that $f$ is monotone, it suffices to show that if the $j$-th bit of $f(w)$ is $1$, and $w \leq u$, then the $j$-th bit of $f(u)$ is $1$, too. Looking into Definition \ref{def:graphgameform}, we see that this just means that if Player $1$ can win from $s_j$ for some assignment of winning and losing leaves, then changing some losing leaves to winning does not change this. This is clear, because whatever winning strategy $\sigma$ Player 1 has in the original configuration will prevent any of the originally losing leaves being reached anyway. Thus, it keeps winning after the modification.
\end{description}
\end{proof}
\end{theorem}

Since the implicit proof of $2.$ implies $1.$ given in Theorem \ref{theo:monotonefunctions} does not provide us with any inclination of the size of the monotone circuit depending on the size of the reachability graph game form we started with, we give an alternate direct proof:

\begin{proposition}
\label{prop:reachabilitycircuittogame}
From any reachability graph game form (with no parameters) we can extract a monotone circuit computing the function it defines, such that the depth of the circuit is bounded by the size if the graph game form.
\begin{proof}
We start with a reachability graph game form. This may contain additional leaves with fixed outcomes, rather than just the designated leaves. It can also have cycles. For each designated vertex $s_i$ we consider the tree-unfolding of the graph. Whenever we reach a vertex in the tree-unfolding that is a duplicate of one of its predecessors, we replace it by a losing leaf. This ensures that we obtain a finite tree for each $s_i$, and put together, a finite forest. This is justified by the fact that if Player 1 can win a reachability game, she can do so without ever visiting a vertex twice.

We then deal with the additional leaves. If a Player 1 vertex $v$ has an edge to a winning leaf, we replace $v$ by a winning leaf. If a Player 1 vertex has an edge $e$ to a losing leaf, but also other outgoing edges, we remove $e$. If all outgoing edges of a Player 1 vertex $v$ go to losing leaves, we replace $v$ by a losing leaf. We perform the dual operations on Player 2 vertices. Both types of operations are repeated until none are applicable anymore, which will have removed all non-designated leaves.

We then merge again all copies of the original designated leaves in the forest. If all copies have been removed, the corresponding input is ignored. Labeling Player 1 vertices by $\vee$ and Player 2 vertices by $\wedge$ gives us a circuit computing the function defined by the initial graph game form.
\end{proof}
\end{proposition}

The direction construction in Proposition \ref{prop:reachabilitycircuittogame} could still produce a circuit which is significantly larger than the original graph game form, at least in the intermediate steps. It seems very likely that graph game forms with reachability objectives are indeed more concise than monotone circuits. The former in particular can express least fixed point operators. Any monotone function has a least fixed point (by Knaster-Tarski), and any slice of a monotone function is monotone again. Thus, given a monotone function $f : \{0,1\}^{m+n} \to \{0,1\}^{m+n}$ we consider $\overline{y} \mapsto \mu \overline{x} \ f(\overline{x},\overline{y}) : \{0,1\}^n \to \{0,1\}^n$ defined as follows: Fix $\overline{y} \in \{0,1\}^n$ and consider the function $F_{\overline{y}} : \{0,1\}^m \to \{0,1\}^m$ where $F_y(\overline{x}) = \pi_{\overline{x}} f(\overline{x},\overline{y})$. Now $F_{\overline{y}}$ is a monotone function, and thus has a least fixed point $\overline{x}_{\overline{y}}$. We set  $(\mu \overline{x} \ f)(\overline{y}) = \overline{x}_{\overline{y}}$.

Starting with a graph game form with reachability objectives that defines $f$, we can obtain a graph game form with reachability objectives of the same size defining $(\mu \overline{x} \ f)$ by connection all leaves corresponding to $\overline{x}$-inputs to the corresponding designated output vertices. To see that this works as intended, note that unraveling who wins where in the new graph game form corresponds to the usual iteration leading up to the least fixed point. It follows in particular that $(\mu \overline{x} \ f)$ is a monotone function itself. Our reasoning here has established:

\begin{proposition}
The class of functions definable by graph game forms of a certain size and with certain winning conditions is closed under the least fixed point operator.
\end{proposition}

While least fixed point operators have received a lot of attention in the broader area (see e.g.~\cite{dawar6}), the combination of fixed point operators and monotone circuits appears to be unexplored. We are pointing out some specific open questions arising from our considerations:

\begin{question}\label{question:consiseness}
\hfill
\begin{enumerate}
\item If $f$ is expressible by a monotone circuit of size $K$, what can be said about the required size of a monotone circuit computing $\mu \overline{x} \ f$?
\item If $f$ is expressible by a graph game form with reachability objectives of size $L$, what can be said about the required size of a monotone circuit computing $\mu \overline{x} \ f$?
\end{enumerate}
\end{question}

For winning conditions other than reachability (or, dually, safety) extracting the monotone circuit can be more complicated. A related question, namely the complexity of the logical definability of the winning regions, was studied for parity games in \cite{dawar5}.

We restrict our attention to functions of type $f : \{0,1\}^m \to \{0,1\}^m$. Let the \emph{repetition number} of such a function be the largest $k$ such that there is $w \in \{0,1\}^m$ with $w, f(w), f(f(w)), \ldots, f^{k-1}(w)$ all being pairwise distinct. We denote it by $\mathrm{rn}(f)$. Clearly, for reasons of cardinality alone we have that:

\begin{observation}
\label{obs:rntrivialbound}
$\mathrm{rn}(f) \leq 2^m$
\end{observation}

If we would ask for a simple cycle rather than a simple path, i.e.~in addition to the requirements above also for $f^k(w) = w$, then for monotone $f$, it follows that all $f^i(w)$, $f^j(w)$ for $i, j < k$, $i \neq j$ must be incomparable w.r.t.~$\leq$. Sperner's theorem (see e.g.~\cite{lubell}) then shows that $k \leq \binom{m}{\lfloor \frac{m}{2} \rfloor}$. A construction in \cite{aracena} shows that this bound is actually attained. The lower bound from \cite{aracena} of course also applies to our repetition numbers, and we immediately obtain:

\begin{corollary}
There is a monotone $f : \{0,1\}^m \to \{0,1\}^m$ with $\mathrm{rn}(f) \geq \binom{m}{\lfloor \frac{m}{2} \rfloor}$.
\end{corollary}

We can turn the function $f$ attaining the bound $\binom{m}{\lfloor \frac{m}{2} \rfloor}$ into a parameterized game with a single parameter which shows that the exponential bounds in Theorem \ref{theo:main} are not entirely avoidable: As explained above, we can turn the monotone circuit computing $f$ into a reachability game. Each designated leaf is turned into a counter-access state, with the green edge leading to the corresponding designated source. We then pick two adjacent words $w$, $f(w)$ from the witnessing cycle. The red edge from the $k$-th former leaf goes to a winning leaf if $f(w)(k) = 1$, and to a losing leaf otherwise. Finally, we add a gadget involving a starting vertex that allows Player 2 to reach exactly those designated sources corresponding to a $1$ in $w$. Now Player 1 can win for some parameter value $N$ if and only if $f^N(f(w)) \geq w$. By choice of $w$, this can only hold if already $f^N(f(w)) = w$, i.e.~$N$ needs to be at least $\binom{m}{\lfloor \frac{m}{2} \rfloor} - 1$.

Some more combinatorial results related to iterated applications of monotone functions are found in \cite{robert}. Better bounds might be obtained by taking into account not only $m$, but the size of the graph game form, too. As mentioned in the introduction, connections to circuit complexity seem likely. To be precise, we are asking:

\begin{question}
\label{question:rn}
If $f : \{0,1\}^m \to \{0,1\}^m$ is realized by a graph game form of size $K$ and objectives of type $C$, how big can $\mathrm{rn}(f)$ be? Particularly relevant are the cases of reachability and parity objectives.
\end{question}

\section{Proving the main theorem}
We now have the ingredients in place to prove our main theorem. For that, we show how to obtain graph game forms from parameterized games; and how the functions defined by those graph game forms are linked to who wins for certain parameter values in the original parameterized games:

\begin{definition}
Given a parameterized game with $N + 1$ counters, the induced graph game form has $N$ counters and is obtained as follows: For each vertex $v$ in $V_c^{N+1}$ we make two copies, $v^{\mathrm{in}}$ and $v^{\mathrm{out}}$. We remove all outgoing edges from $v^{\mathrm{out}}$, and we remove the red outgoing edge from $v^{\mathrm{in}}$, and colour the green outgoing edge of $v^{\mathrm{in}}$ black. We delete all incoming edges to $v^{\mathrm{in}}$, and then designate $v^{\mathrm{in}}$ as a source and $v^{\mathrm{out}}$ as a sink in the graph game form.
\end{definition}

\begin{lemma}
\label{lemma:connection}
Consider parameterized game with $N + 1$ counters, and let $f : \{0,1\}^m \to \{0,1\}^m$ be the function computed by the induced graph game form for fixed values of $n_0^1,\ldots,n_0^N$. Let $w \in \{0,1\}^m$ denote for each $v \in V_c^{N+1}$ whether Player 1 can win from there with the additional parameter value $n_0^{N+1}$. Then $f(w)$ states for each $v \in V_c^{N+1}$ whether Player 1 can win from there with the additional parameter value $n_0^{N+1} + 1$.
\begin{proof}
Let us assume that the $j$-th bit of $f(w)$ is $1$. Consider a strategy $\sigma$ of Player $1$ witnessing this as in Definition \ref{def:graphgameform}. Let $L$ be the set of leaves that Player 2 can reach if Player $1$ follows $\sigma$. Clearly, whenever $\ell_i \in L$, then $w(i) = 1$. If $w(i) = 1$, then Player 1 has a winning strategy $\sigma_i$ starting from the vertex $v_i \in V_c^{N+1}$ that gave rise to $\ell_i$ for parameter values $n_0^1,\ldots,n_0^N, n_0^{N+1}$ by assumption. Now the strategy \emph{follow $\sigma$ until some $\ell_i \in L$ is reached, then switch to following $\sigma_i$} is a winning strategy from $v_j$ for parameter values $n_0^1,\ldots,n_0^N, n_0^{N+1}+1$.

Conversely, if Player 1 has a winning strategy $\sigma$ in the parameterized game $n_0^1,\ldots,n_0^N, n_0^{N+1}+1$ starting from some vertex $v_j \in V_c^{N+1}$, then the prefix of that strategy up to the point where again a vertex in $V_c^{N+1}$ is reached also serves as a witnessing strategy in the graph game form. Any such reachable vertex $v_i \in V_c^{N+1}$ must correspond to a starting vertex from which Player 1 can win with parameters $n_0^1,\ldots,n_0^N, n_0^{N+1}$.
\end{proof}
\end{lemma}

\begin{corollary}
\label{corr:boundingvalue}
Consider parameterized game with $N + 1$ counters, and let $f : \{0,1\}^m \to \{0,1\}^m$ be the function computed by the induced graph game form for fixed values of $n_0^1,\ldots,n_0^N$. Then for any choice of $n_0^{N+1} \in \mathbb{N}$ there exists some $\overline{n}_0^{N+1} \leq \mathrm{rn}(f)$ such that Player $1$ wins the game for $n_0^1,\ldots,n_0^N,n_0^{N+1}$ iff she wins the game for $n_0^1,\ldots,n_0^N,\overline{n}_0^{N+1}$.
\begin{proof}
Let $w$ be chosen as in Lemma \ref{lemma:connection} for $n_0^{N+1} = 0$. By iterative application of Lemma \ref{lemma:connection}, the sequence $w,f(w),f(f(w)),f^3(w),\ldots$ describes the information whether the game is won for Player 1 for $n_0^{N+1} = 0$, $n_0^{N+1} = 1$, and so on. By definition of the repetition number, any potential value that appears in this sequence already appears at an index $\overline{n}_0^{N+1} \leq \mathrm{rn}(f)$.
\end{proof}
\end{corollary}

\begin{proof}[Proof of Theorem \ref{theo:main}]
We can apply Corollary \ref{corr:boundingvalue} to each counter individually (by reordering them). Observation \ref{obs:rntrivialbound} provides the concrete bounds.
\end{proof}

\section{Generic Application}
\label{sec:application-generic}
It is often a convenient technique to reduce a decision problem to the existence of a winning strategy in some game, and then use known decidability and complexity results for the latter as upper bounds for the former. The games constructed in such reductions often are intuitively accessible, and have an interpretation of both players making choices trying to either prove or disprove the statement.

A simple yet fundamental example is the universality problem of non-deterministic parity automata mentioned in the introduction: From a parity automaton we can move to a parity game, where the first player controls the input to the automaton, and the other player controls how non-determinism is resolved. The winning condition for Player 2 in the game is the acceptance condition of the original automaton. If there is some $\omega$-word $\rho$ which is not accepted by the automaton, then if Player 1 plays according to $\rho$, Player 2 cannot win against this. Conversely, if Player 2 has a winning strategy, then the automaton accepts all words. Deciding who wins parity games is a well-studied computational problem, of course (e.g.~\cite{jurdzinski,khoussainov}). A crucial observation is that parity games are positionally determined, which makes the required reasoning much simpler -- and this observation is less directly accessible if we were to reason about parity automata directly.

Now the generic form of an application of our main theorem proceeds as follows: Assume that we can reduce deciding some property $\phi$ of some type of structure $M$ to the existence of winning strategies in a derived game $G(M)$. The construction of $G(M)$ tells us how we can define \emph{parameterized structures} $M$ in such that a way that for a parameterized structure $M(n_0,\ldots,n_k)$ the construction yields a parameterized game $G(M(n_0,\ldots,n_k))$. Now our main theorem tells us that any quantified statement like $\Phi(M) := \forall n_0 \forall n_1 \exists n_2 \ \ldots \forall n_{k-1} \exists n_k \ \phi(M(n_0,\ldots,n_k))$ is equivalent to the bounded version $\forall n_0 \leq N \forall n_1 \leq N \exists n_2 \leq N  \ \ldots \forall n_{k-1} \leq N \exists n_k \leq N \ \phi(M(n_0,\ldots,n_k))$, where $N$ is at most exponential in the size of $G(M)$. If deciding the winner in games of type $G(M)$ takes time $T(n)$, then it follows that formula like $\Phi$ can be decided in time $2^{nk}T(2^n)$. Improved bounds in Question \ref{question:rn} would improve this time bound.

If we apply the generic form to the fundamental example of non-deterministic automata (of some type), we arrive at the notion of a parameterized automaton as defined in \cite[Definition 32]{ba-arxiv}. This is just an automaton with counters starting at some initial values, which can be decremented and tested for $0$. Counter automata are of course a mainstay of automata theory (e.g.~\cite{valiant,boehm}), but the similarities are superficial: Counter automata enhance the usual capabilities of automata by allowing access to counter which start at $0$, and are incremented, decremented and tested for $0$. Since the counter in a parameterized automaton can never be incremented, they do not confer any additional computational power. Moreover, it is essential for the notion to make sense that the counter values start at some parameter values. Parameterized automata thus do not recognize individual languages, but rather families of languages indexed by the parameter values.

The types of questions that our result shows to be decidable for such uniform families of regular (or $\omega$-regular) languages include e.g.: $$\forall n \in \mathbb{N} \ \exists k \in \mathbb{N} \ \exists j \in \mathbb{N} \ \ L_n \subseteq M_k \wedge M_k \subseteq L_{n+j}$$ where $(L_n)_{n \in \mathbb{N}}$ and $(M_n)_{n \in \mathbb{N}}$ are families of languages coded by parameterized automata. We could have significantly more involved quantifier structures, and also involve operations on regular languages that are witnessed by automata constructions, such as taking products, intersections, Kleene-star, and so on.

\section{Application to chains of strategies}
\label{sec:application-chains}
Our motivating application for parameterized games is to study chains of strategies in the context of dominance between strategies. We recall:

\begin{definition}
Consider a game with payoff functions $\pi$ for Player $1$. A strategy $\sigma$ of Player $1$ is \emph{dominated} by an alternative strategy $\sigma'$, if for all strategies $\tau$ of Player $2$, we find that $\pi(\Run(\sigma,\tau)) \leq \pi(\Run(\sigma',\tau))$ and there exists a strategy $\tau'$ of Player $2$ such that $\pi(\Run(\sigma,\tau')) < \pi(\Run(\sigma',\tau'))$. A strategy $\sigma$ which is not dominated by any other strategy is called \emph{admissible}.
\end{definition}

A rational player should always be willing to change her strategy to another one dominating her original choice. If every strategy is dominated by an admissible strategy, then the admissible strategies can be seen as the rational choices. An approach to reactive synthesis based on this concept was proposed in \cite{BrenguierRS17}.

The investigation of dominance and admissibility in games played on finite graphs was initiated in \cite{Berwanger07} and subsequently became an active research topic (e.g.~\cite{BRS14,BrenguierPRS16,BGRS17,pauly-raskin2,DBLP:conf/birthday/BassetRS17}). A central result in \cite{Berwanger07} is that for boolean objectives (i.e.~games where the payoff function has range $\{0,1\}$) and perfect information, every strategy is either admissible or dominated by an admissible one. \cite[Example 1]{ba-arxiv} (reproduced here as Figure \ref{fig:best_animal}) shows that already for three distinct payoffs this theorem fails; and \cite[Example 9]{BrenguierPRS17} (reproduced below) shows that for reachability games with imperfect information, the theorem fails, too.

\begin{example}[{\cite[Example 9]{BrenguierPRS17}}]
\label{ex:noadmissible}
We exhibit a regular game with a reachability objective (to reach the state marked $w$) lacking observation-based admissible strategies for Player $1$. The graph is depicted in Figure \ref{fig:ex-noadm}.  Player $1$ controls circle vertices, Player $2$ controls box vertices. Player $1$ wins all plays that eventually enter the vertex $w$. Player 1 has imperfect information: he can not differentiate between $s_1$ and $s_2$, and not between $t_1$ and $t_2$. As a consequence, the observation-based strategies available to Player $1$ are essentially the strategies $\sigma_n$, one for each $n \in \mathbb{N}$, which play the action "$0$" for $n$ consecutive steps followed by the action "$1$", and $\sigma_\infty$, which always plays the action "$0$". Then it is easy to show that $\sigma_\infty \prec \sigma_0 \prec \sigma_1 \ldots$, hence there is no observation-based admissible strategy for Player~1 in this game.

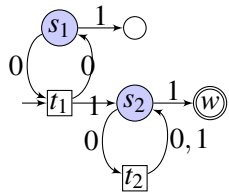
\begin{figure}[htbp]
  \begin{center}
    \begin{tikzpicture}[xscale=0.5,yscale=0.5,minimum size=0.3cm, inner
      sep=1pt]
      \path (0,0) node[draw,rectangle, label=below:$ $] (a) {$t_1$};
      \path (0,2) node[draw,circle,fill=blue!20, label=below:$ $] (b) {$s_1$};
      \path (2,2) node[draw,circle, label=below:$ $] (c) {$ $};
      \path (2,0) node[draw,circle,fill=blue!20, label=below:$ $] (d) {$s_2$};
      \path (2,-2) node[draw,rectangle, label=below:$ $] (e) {$t_2$};
      \path (4,0) node[draw,circle,accepting, label=below:$ $] (f) {$w$};

      \draw[arrows=-latex'] (-1,0) -- (a) ;

      \draw[arrows=-latex'] (a) .. controls +(20:30pt) and +(-20:30pt) .. (b) node[pos=.5,right=-0.2cm] {$0$};
      \draw[arrows=-latex'] (b) .. controls +(20:-30pt) and +(-20:-30pt) .. (a) node[pos=.5,left] {$0$};
      \draw[arrows=-latex'] (b) -- (c) node[pos=.5,above] {$1$};
      \draw[arrows=-latex'] (a) -- (d) node[pos=.5,above=-0.2cm] {$1$};
      \draw[arrows=-latex'] (e) .. controls +(20:30pt) and +(-20:30pt) .. (d) node[pos=.5,right] {$0,1$};
      \draw[arrows=-latex'] (d) .. controls +(20:-30pt) and +(-20:-30pt) .. (e) node[pos=.5,left] {$0$};
      \draw[arrows=-latex'] (d) -- (f) node[pos=.5,above] {$1$};
    \end{tikzpicture}
    \caption{Player 1 can not differentiate between $s_1$ and $s_2$, and not between $t_1$ and $t_2$. In this example, no strategy is admissible for Player 1.}
    \label{fig:ex-noadm}
  \end{center}
\end{figure}
\end{example}

In both the quantitative (i.e.~more than two distinct payoffs) and the imperfect information case, we see two particular features arising:

\begin{enumerate}
\item The counterexamples are built around situations where repeating a certain action (i.e.~choosing a certain edge at some vertex in the graph) more and more often produces strategies that dominate the one repeating that action less often; however; always choosing that action does not dominate choosing it finitely often, and in fact, may be strictly worse. Trying to formalize the idea of \emph{repeat this action a large but finite number of times} yields the notion of a parameterized strategy, which can be realized by a parameterized automaton (cf.~Section \ref{sec:application-generic}).
\item We can decide whether one strategy dominates another, if both are regular and given via automata realizing them. Moreover, this is proven by constructing a game based on the original game and the two strategies involved, in such a way that a specific player has a winning strategy in the derived game if and only if the dominance holds. For the imperfect information case, this is done in \cite[Section 7]{BrenguierPRS17}; for the quantitative case, this is \cite[Lemma 39]{ba-arxiv}.
\end{enumerate}

The first observation formed the motivation in \cite{ba-arxiv} to consider \emph{increasing chains of strategies} rather than single strategies as the unit of analysis. Under mild constraints (restricting consideration to some countable set of strategies, such as only computable or regular ones), it is shown that any chain of strategies is maximal (w.r.t.~dominance) or itself dominated by a maximal chain. The argument is non-constructive, and proceeds via Zorn's Lemma. A more concrete approach is to study chains realized by parameterized automata, dubbed \emph{uniform chains}. In generalized safety/reachability games, these suffice indeed to ensure that every strategy is admissible or dominated by a maximal uniform chain. Of course, one can also consider uniform chains in games with imperfect information.

An obvious algorithmic question arises, namely to decide whether one uniform chain is dominated by another. Applying the construction that reduces dominance between (single) strategies to the existence of winning strategies in a derived game to uniform chains, we obtain a parameterized game with two parameters $n_0, n_1$. Dominance between the uniform chains then corresponds to asking whether for all $n_0$ there exists an $n_1$ such that Player 1 wins the corresponding instantiation of this parameterized game. In \cite{ba-arxiv}, decidability (in polynomial time) of this question is shown by a lengthy argument that essentially provides bounds for how large values of the parameters need to be considered. The arguments in particular make use of the specific winning conditions considered there, namely generalized safety/reachability games.

As alternative proof, we can conclude the existence of such bounds from our main theorem. While the bounds obtained such are worse than the ones obtained \emph{by hand} in \cite{ba-arxiv} (and in particular do not establish decidability in polynomial time), the restriction to generalized safety/reachability games is unnecessary for this. In fact, since \cite[Lemma 39]{ba-arxiv} holds (with small adjustments) for arbitrary $\omega$-regular conditions, we immediately get:

\begin{corollary}
Given a perfect-information game with finitely many $\omega$-regular payoffs, and two uniform chains of strategies, it is decidable whether the first is dominated by the second.
\end{corollary}

Similarly, we can start with \cite[Lemma 34]{BrenguierPRS17} to obtain:

\begin{corollary}
Given a boolean imperfect information game, and two uniform chains of strategies, it is decidable whether the first is dominated by the second (w.r.t.~arbitrary strategies used by the opponent).
\end{corollary}

\section{Conclusion}
We have introduced the formalism of parameterized games, which lets us reason about parameterized automata as introduced in \cite{ba-arxiv}. These notions define uniform families of various types of \emph{regular objects} (such as languages, strategies, etc) in such a way that quantifying in an arbitrary way over the parameter preserves decidability of properties. We discussed how to apply this to chains of strategies in the context of admissibility and dominance between strategies, but other potential applications seem likely.

The proof of our main result involved the notion of graph game forms, and the monotone boolean functions defined by them. Several open question remain (Question \ref{question:consiseness} and Question \ref{question:rn}) regarding the conciseness of the framework, and regarding how quickly functions defined such converge to their least fixed point.

\section*{Acknowledgements}
I am grateful to Isma\"el Jecker and Marie Van den Bogaard for discussion leading to the inception of this work. I received helpful comments from the anonymous referees, Anuj Dawar and Mickael Randour.

\bibliographystyle{eptcs}
\bibliography{biblio}

\newpage
\appendix
\section*{Additional remarks}

The following is a straight-forward adaption of the familiar construction of an undetermined game to prefix-independent winning conditions and games played on finite graphs.

\begin{proposition}
\label{prop:undetermined}
There exists a (two-player win/lose perfect information) game played on a finite graph with prefix-independent winning condition, such that neither player has a winning strategy.
\begin{proof}
The arena looks as follows: The game starts in $V_s$, controlled by Player 1 with outgoing edges to $V_{a,0}$ and $V_{a,1}$. Both of these are controlled by Player 2 and have outgoing edges to $V_{b,0}$ and $V_{b,1}$. The latter two have only a single outgoing edge leading to $V_s$. The important property of this arena is that for any fixed strategy of either player, there are still continuumsly many runs compatible with it.

Next, we consider the quotient of $V^\omega$ by the equivalence relation $E$ where $(p,q) \in E$ iff $\exists n, m \ p_{\geq n} = q_{\geq m}$. Note that each equivalence class from $E$ is countable, $V^\omega/E$ has again cardinality continuum. The same holds for any subset of $V^\omega$ of size continuum.

Let $\Omega$ be the least ordinal of size continuum, and let $(\sigma_\alpha)_{\alpha < \Omega}$ and $(\tau_\beta)_{\tau < \Omega}$ be well-orderings of the strategies of Player 1 and 2 respectively.

We now construct the winning set $C$ and its complement $V^\omega \setminus C$ in stages $\gamma < \Omega$, i.e.~we built sets $(C_\gamma)_{\gamma < \Omega}$ and $(D_\gamma)_{\gamma < \Omega}$ such that $C_\gamma \subseteq C_{\gamma'}$ and $D_\gamma \subseteq D_{\gamma'}$ for $\gamma < \gamma'$ such that $C = \bigcup_{\gamma < \Omega} C_\gamma$, $V^\omega \setminus C \supseteq (D_\gamma)_{\gamma < \Omega}$, and $|C_\gamma| < 2^{\aleph_0}$, $|D_\gamma| < 2^{\aleph_0}$.

We start with $C_0 = D_0 = \emptyset$. In stage $\gamma$, pick the least $\beta$ such that $\Run(\sigma_\gamma,\tau_\beta) \notin \bigcup_{\gamma' < \gamma} C_\gamma'$. This must exist, since the set $\{\Run(\sigma_\gamma,\tau_\beta) \mid \beta < \Omega\}$ has cardinality continuum, whereas $\bigcup_{\gamma' < \gamma} C_\gamma'$ is smaller. Now let $D_\gamma$ be the union of the $E$-equivalence class of $\Run(\sigma_\gamma,\tau_\beta)$ and $\bigcup_{\gamma' < \gamma} D_\gamma$. Next, pick the smallest $\alpha$ such that $\Run(\sigma_\alpha,\tau_\gamma) \notin \bigcup_{\gamma' \leq \gamma} D_\gamma'$, which exists as before. Let $C_\gamma$ be the union of the $E$-equivalence class of $\Run(\sigma_\alpha,\tau_\gamma)$ and $\bigcup_{\gamma' < \gamma} C_\gamma$.

As we have designed the winning condition to be $E$-invariant, it is prefix-independent. For every strategy of either player, we ensured that there exists a strategy of their opponent to beat it, hence no player can have a winning condition.
\end{proof}
\end{proposition}
\end{document}